\documentclass[a4paper,11pt]{amsart}
  
\usepackage[utf8]{inputenc}

\usepackage[colorinlistoftodos,bordercolor=orange,backgroundcolor=orange!20,linecolor=orange,textsize=tiny]{todonotes}

\usepackage{amsmath}  
\usepackage{amssymb}     
\usepackage{bbm}   
   
\usepackage{bm}
\usepackage{hyperref}
\usepackage{mathrsfs}
\usepackage{mathtools} 
\usepackage{multirow}
\usepackage{paralist}
\usepackage{phdalgo}
\usepackage{fixmath}
\usepackage{fullpage}

\usepackage{tikz}
\usetikzlibrary{arrows}
\usetikzlibrary{patterns}
\usetikzlibrary{calc}
\usetikzlibrary{shapes}
\usetikzlibrary{positioning}

\newcommand{\arxiv}[1]{arXiv:#1}

\newcommand{\ie}{\textit{i.e.}}
\newcommand{\eg}{\textit{e.g.}}

\newcommand{\etal}{\textit{et al.}}

\DeclareMathOperator{\sign}{\mathrm{sign}}
\DeclareMathOperator{\tsign}{\mathrm{tsign}}
\DeclareMathOperator{\tdet}{\mathrm{tdet}}

\DeclareMathAlphabet{\mathbfcal}{OMS}{cmsy}{b}{n}

\DeclareMathAlphabet{\mathbbold}{U}{bbold}{m}{n}

\newcommand{\R}{\mathbb{R}}

\newcommand{\puiseux}{\mathbb{K}}

\newcommand{\trop}[1][]{\ifthenelse{\equal{#1}{}}{ \mathbb{T} }{ \mathbb{T}(#1) }}
\newcommand{\strop}[1][]{{\trop[#1]}_{\pm}}

\newcommand{\transpose}[1]{#1^{\intercal}}

\newcommand{\tplus}{\oplus}  
\newcommand{\tsum}{\bigoplus}
\newcommand{\tdot}{\odot}
\newcommand{\tprod}{\bigodot}
\newcommand{\tminus}{\ominus}  

\newcommand{\zero}{\mathbbold{0}} 
\newcommand{\unit}{\mathbbold{1}} 
\newcommand{\Pcal}{\mathcal{P}}

\newcommand{\Ecal}{\mathcal{E}}

\DeclareMathOperator*{\val}{val}

\DeclareMathOperator*{\sval}{sval}

\newcommand{\Pbcal}{\mathbfcal{P}}

\newcommand{\Ebcal}{\mathbfcal{E}}  
\newcommand{\bo}[1]{\mathbold{#1}}

\DeclareMathOperator{\tropPol}{\mathrm{trop}}
\DeclareMathOperator{\new}{\mathrm{New}}

\newcommand{\eps}{\epsilon}

\newcommand{\PCBC}{\Call{PCBC}{}}
\newcommand{\tropPCBC}{\Call{TropPCBC}{}}

\newcommand{\sigmagstrop}{\rho^{\mathrm{trop}}_{\mathrm{sv}}}
\newcommand{\sigmags}{\rho_{\mathrm{sv}}}

\newcommand{\boid}{\mathbf{Id}}
\newcommand{\disjcup}{\uplus}
\newcommand{\compl}[1]{\overline{#1}}

\newcommand{\None}{\mathsf{None}}
\newcommand{\vout}{k_{\mathrm{out}}}
\newcommand{\vin}{k_{\mathrm{in}}}

\DeclareMathOperator{\tropLP}{\mathrm{LP}}
\DeclareMathOperator{\LP}{\mathbf{LP}}
\DeclareMathOperator{\diag}{\mathrm{diag}}
\DeclareMathOperator{\tdiag}{\mathrm{tdiag}}
\DeclareMathOperator{\idx}{\mathrm{idx}}

\makeatletter
\newcommand{\raisemath}[1]{\mathpalette{\raisem@th{#1}}}
\newcommand{\raisem@th}[3]{\raisebox{#1}{$#2#3$}}
\makeatother

\newtheorem{theorem}{Theorem}
\newtheorem{proposition}[theorem]{Proposition}

\newtheorem{lemma}[theorem]{Lemma}
\theoremstyle{definition}

\newtheorem{assumption}{Assumption}

\theoremstyle{remark}
\newtheorem{remark}[theorem]{Remark}
\newtheorem{example}[theorem]{Example}

\tikzset{grid/.style={gray!30,very thin}}
\tikzset{axis/.style={gray!50,->,>=stealth'}}
\tikzset{convex/.style={draw=none,fill=lightgray,fill opacity=0.7}}
\tikzset{convexborder/.style={very thick}}
\tikzset{point/.style={blue!50}}
\tikzset{hs/.style={fill opacity=0.3,fill=orange,draw=none}}
\tikzset{hsborder/.style={orange,ultra thick,dashdotted}}

\title{The tropical shadow-vertex algorithm solves mean payoff games in polynomial time on average}
\date{\today}

\thanks{X.~Allamigeon and S.~Gaubert are partially supported by the PGMO program of EDF and Fondation Math\'ematique Jacques Hadamard. P.~Benchimol is supported by a PhD fellowship of DGA and \'Ecole Polytechnique.}

\author{Xavier {A}llamigeon}
\author{Pascal Benchimol}
\author{St{\'e}phane {G}aubert}

\address{INRIA and CMAP, \'Ecole Polytechnique, CNRS, 91128 Palaiseau Cedex France}
\email{firstname.lastname@inria.fr}

\begin{document}
\begin{abstract}
We introduce an algorithm which solves mean payoff games in polynomial time on average, assuming the distribution of the games satisfies a flip invariance property on the set of actions associated with every state. The algorithm is a tropical analogue of the shadow-vertex simplex algorithm, which solves mean payoff games via linear feasibility problems over the tropical semiring $(\R \cup \{-\infty\}, \max, +)$. The key ingredient in our approach is that the shadow-vertex pivoting rule can be transferred to tropical polyhedra, and that its computation reduces to optimal assignment problems through Pl\"ucker relations. 
\end{abstract}

\maketitle

\section{Introduction}
 
A \emph{mean payoff game} involves two opponents, ``Max'' and ''Min'', who alternatively move a pawn over the nodes of a weighted bipartite digraph. The latter consists of two classes of nodes, represented by squares and circles, and respectively indexed by $i \in [m]$ and $j \in [n]$ (we use the notation $[k] := \{1, \dots, k\}$). The weight of the arc $(i,j)$ (resp.\ $(j,i$)) is a real number denoted by $A_{ij}$ (resp.\ $B_{ij}$). We set $A_{ij} := -\infty$ (resp.\ $B_{ij} := -\infty$) when there is no such arc. An example of game
is given in Figure~\ref{fig:game}.

When the pawn is placed over a square node $i$, Player Max selects an outgoing arc $(i,j)$, and then moves the pawn to circle node $j$ and receives the payment $A_{ij}$ from Player Min. Conversely, when the pawn is located on a circle node $j$, Player Min chooses an arc $(j,i')$, moves the pawn to square node $i'$, and Player Max pays her the amount $B_{i'j}$. We assume that $A$ (resp.\ $B$) does not have any identically $-\infty$ row (resp.\ column), so that both players have at least one possible action at every node. The game starts from a circle node $j_0 = j$, and then the two players make infinitely many moves, visiting a sequence $j_0, i_1, j_1, i_2, \dots$ of nodes. The objective of Player Max is to maximize his mean payoff, defined as the liminf of the following ratio when $p \rightarrow + \infty$: 
\begin{equation}
 (-B_{i_1 j_0} + A_{i_1 j_1} - B_{i_2 j_1} + A_{i_2 j_2} + \dots - B_{i_p j_{p-1}} + A_{i_p j_p})/p \ . \label{eq:ratio}
\end{equation}
Symmetrically, Player Min aims at minimizing 
her
mean loss, defined as the limsup of~\eqref{eq:ratio} when $p \rightarrow + \infty$. 
Mean payoff games can be defined more generally
over arbitrary (not necessarily bipartite)
digraphs.
This situation can be reduced to the present one~\cite{ZwickPaterson1996}.

\begin{figure}[t]
\vskip-15ex
\begin{center}
\begin{tikzpicture}[scale=0.9,>=stealth',row/.style={draw,rectangle,minimum size=0.5cm},col/.style={draw,circle,minimum size=0.5cm}]
\node[row] (i1) at (0,0) {$1$};
\node[row] (i2) at (3,0) {$2$};
\node[row] (i3) at (6,0) {$3$};

\node[col] (j1) at (0,2.5) {$1$};
\node[col] (j2) at (3,2.5) {$2$};
\node[col] (j3) at (6,2.5) {$3$};

\draw[->] (i1) to[out=120,in=-120] node[above left=0ex and -0.5ex,font=\small] {$5$} (j1);
\draw[->] (i1) -- node[above right=2ex and -0.5ex,font=\small] {$-2$} (j2);

\draw[->] (j1) to[out=-60,in=60] node[below left=0ex and -0.3ex,font=\small] {$-1$} (i1);
\draw[->] (j1) -- node[below right=1.5ex and 0.5ex,font=\small] {$3$} (i2);

\draw[->] (i2) to[out=20,in=-120] node[above right=2ex and 3.5ex,font=\small] {$2$} (j3);
\draw[->] (j2) to[out=-90,in=90] node[below right=0ex and -0.5ex,font=\small] {$10$} (i2);
\draw[->] (j2) to[out=-20,in=120] node[above=0.5ex,font=\small] {$1$} (i3);

\draw[->] (j3) .. controls (-2.5,5) and (-2,1) .. node[left,font=\small] {$3$} (i1);
\draw[->] (i3) to[out=160,in=-60]  node[above left=1ex and -0.5ex,font=\small] {$0$} (j2);
\end{tikzpicture}
\end{center}
\caption{An example of mean payoff game. Circle node $1$ is a winning initial state for Player Max, while circle node $2$ is losing.} \label{fig:game}
\end{figure}
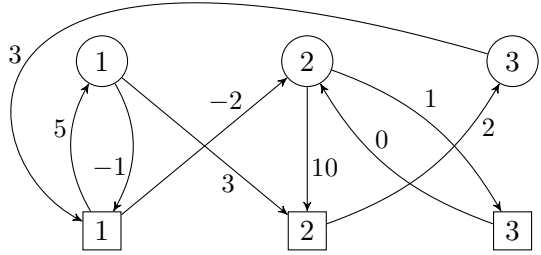

Mean payoff games were
studied by Ehrenfeucht and Mycielski~in \cite{EhrenfeuchtMycielski1979}, 
where they
proved
that these games have a value and positional optimal strategies. In more detail, for every initial state $j \in [n]$, there exists a real $\chi_j$ and  positional strategies $\sigma : [m] \rightarrow [n]$ and $\tau : [n] \rightarrow [m]$, such that:
\begin{inparaenum}[(i)]
\item Player Max is certain to win a mean payoff greater than or equal to $\chi_j$ with the strategy $\sigma$ (\ie~by choosing the arc $(i,\sigma(i))$ every time the pawn is on a square node $i \in [m]$), 
\item Player Min is sure that
her
mean loss is less than or equal to $\chi_j$ with the strategy $\tau$.
\end{inparaenum}

The decision problem associated with mean payoff games consists in determining whether the initial state $j$ is \emph{winning} for Player Max, \ie~$\chi_j \geq 0$. 
The question
of the existence of
a polynomial time algorithm solving this problem 
was
first raised by Gurvich, Karzanov and Khachiyan in~\cite{GurvichCMMP1988}. 
This problem was shown
to be in $\text{NP} \cap \text{co-NP}$ by Zwick and Paterson in~\cite{ZwickPaterson1996}. 
While mean payoff games (and the related class of parity games) received an important attention over the past years~\cite{GurvichCMMP1988,ZwickPaterson1996,Jurdzinski1998,GaubertGunawardena98,VogeJurdzinski2000,BjorklundVorobyovDAM2007,JurdzinskiPatersonZwick2008,Friedmann2009,Briem2011}, all the algorithms developed so far are superpolynomial, and the question raised by Gurvich~\etal{} is still open.

The present work exploits the equivalence between mean payoff games and linear feasibility problems in tropical algebra. Indeed, it was shown by Akian, Gaubert and Guterman in~\cite{AkianGaubertGuterman2012} that the initial state $n$ is winning for Player Max in the game with payments matrices $A,B$ if, and only if, there exists a solution $x \in (\R \cup \{-\infty\})^{n-1}$ to the following system of inequalities: 
\begin{multline}
\forall i \in [m], \qquad \max(A_{i1} + x_1, \dots, A_{i(n-1)} + x_{n-1}, A_{in}) \\
\geq \max(B_{i1} + x_1, \dots, B_{i(n-1)} + x_{n-1}, B_{in}) \ . \label{eq:trop_ineq}
\end{multline}
The constraints of the form~\eqref{eq:trop_ineq} correspond to affine inequalities over the tropical (max-plus) semiring, \ie~the set $\trop := \R \cup \{-\infty\}$ endowed with the operations $x \tplus y := \max(x,y)$ as addition, and $x \tdot y := x+y$ as multiplication. The conjunction of finitely many such inequalities defines a \emph{tropical polyhedron}. Solving a mean payoff game consequently amounts to determine whether a tropical polyhedron is empty, which can be thought of as the tropical analogue of the feasibility problem in linear programming. 
This is among the motivations leading to the development of
a tropical simplex method 
in a previous work of the authors and Joswig~\cite{AllamigeonBenchimolGaubertJoswig2013}. Then, complexity results known for the classical simplex algorithm can be potentially transferred to the tropical setting.
However, the main obstacle is to ``tropicalize'' the pivoting rule involved, \ie~to define a tropical pivoting rule which is both compatible with the classical one, and computable, if possible, with a reasonable time complexity. So far~\cite{AllamigeonBenchimolGaubertJoswig2014}, the only pivoting rules which have been tropicalized are \emph{combinatorial}, \ie~they are defined in terms of the neighborhood of the current basic point in the vertex/edge graph of the polyhedron.

\subsection*{Contributions.} We prove that the shadow-vertex simplex algorithm can be tropicalized. Following the average-case analysis of the shadow-vertex algorithm of Adler, Karp and Shamir in~\cite{AdlerKarpShamir87}, we deduce that the tropical algorithm solves mean payoff games in polynomial time on average (Section~\ref{sec:pcbc}). The complexity bound holds when the distribution of the games satisfies a \emph{flip invariance} property. The latter requires that the distribution of the games is left invariant by every transformation consisting, for an arbitrary node of the game, in flipping the orientation of all the arcs incident to this node.  Equivalently, the probability distribution on the set of payment matrices $A,B$ is invariant by every transformation consisting in swapping the $i$th row of $A$ with the $i$th row of $B$, or the $j$th column of $A$ with the $j$th column of $B$.
Figure~\ref{fig:flip} provides the illustration of a discrete distribution of games satisfying the property. 

The key difficulty in our approach is to show that the computation of the tropical shadow-vertex pivoting rule can be done in polynomial time (Section~\ref{sec:shadow_vertex}). To this end, we exploit the fact that the shadow-vertex rule is semi-algebraic, \ie~it is defined in terms of the signs of finitely many polynomials. Under some genericity conditions, we deduce that the tropical rule reduces to classical linear programs over some Newton polytopes, which are actually (Minkowski sums of) bipartite perfect matching polytopes.

\subsection*{Related work.} We are not aware of other works with such average-case complexity results on mean payoff games. In~\cite{RothSODA10}, Roth~\etal{} made a probabilistic analysis of $n \times n$ bi-matrix games with weights chosen independently uniformly in $[0,1]$. Under this assumption, they showed that with high probability (greater than $1 - f(n)$, with $f(n) = o(1/n^c)$ for all constant $c$), such games admit a pure stationary strategy equilibrium parametrized by only $4$ actions. The latter can be consequently found in polynomial time. While this result indicates that complex instances of games are rare, it does not seem to us that it can be used to deduce an average-case complexity bound.

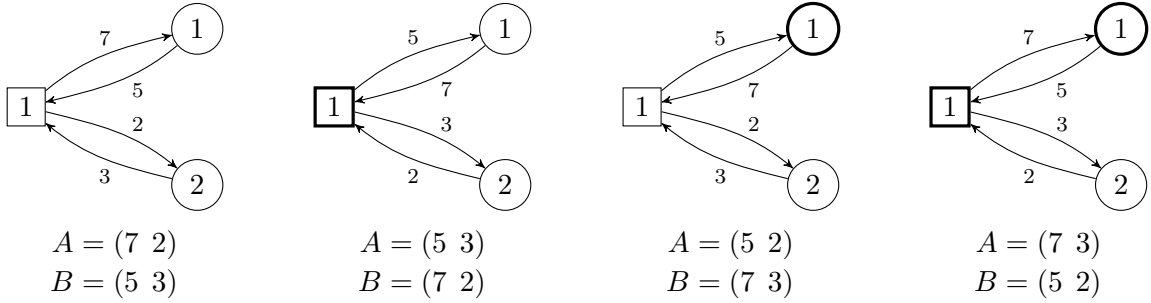
\begin{figure}[t]
\begin{center}
\begin{tikzpicture}[scale=0.9,>=stealth',row/.style={draw,rectangle,minimum size=0.5cm},col/.style={draw,circle,minimum size=0.5cm}]

\begin{scope}
\node[row] (i1) at (0,0) {$1$};

\node[col] (j1) at (2.5,1.15) {$1$};
\node[col] (j2) at (2.5,-1.15) {$2$};

\draw[->] (i1) to[out=40,in=-166] node[above,font=\scriptsize] {$7$} (j1);
\draw[->] (j1) to[out=-140,in=14] node[below right=-0.75ex and 0.5ex,font=\scriptsize] {$5$} (i1);
\draw[->] (i1) to[out=-14,in=140] node[above right=-0.75ex and 0.5ex,font=\scriptsize] {$2$} (j2);
\draw[->] (j2) to[out=166,in=-40] node[below,font=\scriptsize] {$3$} (i1);
\node at (1.25,-2.35) 
{
$
\begin{aligned}
A & = (7 \; \, 2)\\[-0.5ex]
B & = (5 \; \, 3)
\end{aligned}
$
};
\end{scope}

\begin{scope}[xshift=4.5cm]
\node[row, very thick] (i1) at (0,0) {$1$};

\node[col] (j1) at (2.5,1.15) {$1$};
\node[col] (j2) at (2.5,-1.15) {$2$};

\draw[->] (i1) to[out=40,in=-166] node[above,font=\scriptsize] {$5$} (j1);
\draw[->] (j1) to[out=-140,in=14] node[below right=-0.75ex and 0.5ex,font=\scriptsize] {$7$} (i1);
\draw[->] (i1) to[out=-14,in=140] node[above right=-0.75ex and 0.5ex,font=\scriptsize] {$3$} (j2);
\draw[->] (j2) to[out=166,in=-40] node[below,font=\scriptsize] {$2$} (i1);
\node at (1.25,-2.35) 
{
$
\begin{aligned}
A & = (5 \; \, 3)\\[-0.5ex]
B & = (7 \; \, 2)
\end{aligned}
$
};
\end{scope}

\begin{scope}[xshift=9cm]
\node[row] (i1) at (0,0) {$1$};

\node[col, very thick] (j1) at (2.5,1.15) {$1$};
\node[col] (j2) at (2.5,-1.15) {$2$};

\draw[->] (i1) to[out=40,in=-166] node[above,font=\scriptsize] {$5$} (j1);
\draw[->] (j1) to[out=-140,in=14] node[below right=-0.75ex and 0.5ex,font=\scriptsize] {$7$} (i1);
\draw[->] (i1) to[out=-14,in=140] node[above right=-0.75ex and 0.5ex,font=\scriptsize] {$2$} (j2);
\draw[->] (j2) to[out=166,in=-40] node[below,font=\scriptsize] {$3$} (i1);
\node at (1.25,-2.35) 
{
$
\begin{aligned}
A & = (5 \; \, 2)\\[-0.5ex]
B & = (7 \; \, 3)
\end{aligned}
$
};
\end{scope}

\begin{scope}[xshift=13.5cm]
\node[row, very thick] (i1) at (0,0) {$1$};

\node[col, very thick] (j1) at (2.5,1.15) {$1$};
\node[col] (j2) at (2.5,-1.15) {$2$};

\draw[->] (i1) to[out=40,in=-166] node[above,font=\scriptsize] {$7$} (j1);
\draw[->] (j1) to[out=-140,in=14] node[below right=-0.75ex and 0.5ex,font=\scriptsize] {$5$} (i1);
\draw[->] (i1) to[out=-14,in=140] node[above right=-0.75ex and 0.5ex,font=\scriptsize] {$3$} (j2);
\draw[->] (j2) to[out=166,in=-40] node[below,font=\scriptsize] {$2$} (i1);
\node at (1.25,-2.35) 
{
$
\begin{aligned}
A & = (7 \; \, 3)\\[-0.5ex]
B & = (5 \; \, 2)
\end{aligned}
$
};
\end{scope}
\end{tikzpicture}
\end{center}
\caption{A distribution of game satisfying the flip invariance property (with $m = 1$ and $n = 2$), together with the payment matrices. The four configurations are supposed to be equiprobable. The nodes on which the flip operations have been performed are depicted in bold.}\label{fig:flip}
\end{figure}
 
\section{Preliminaries}\label{sec:preliminaries}

\subsection{Tropical arithmetic and generalized Puiseux series}

As previously discussed, the triple $(\trop, \tplus, \tdot)$ forms a semiring, and  
the elements $\zero := -\infty$ and $\unit := 0$ correspond to the zero and unit respectively. 
The tropical operations can be extended to matrices with entries in $\trop$ in a usual way, by defining $A \tplus B := (A_{ij} \tplus B_{ij})_{ij}$ and $A \tdot B := (\tsum_k A_{ik} \tdot B_{kj})_{ij}$. We also introduce the exponentiation $x^{\tdot k}$ for any $x \in \trop$ and $k \in \mathbb{N}$, which is defined as the product $x \tdot x \tdot \dots \tdot x$ of $k$ occurrences of $x$ (if $k = 0$, it is set to $\unit$). 

Even if the addition $\tplus$ does not have an inverse, it is convenient to consider tropical numbers with a negative sign. The sign is encoded in a formal way, by introducing two copies $\trop_+$ and $\trop_-$ of $\trop \setminus \{\zero\}$, respectively consisting of the positive and negative elements. The set $\strop$ of \emph{tropical signed numbers} is defined as $\trop_+ \cup \trop_- \cup \{\zero\}$. The elements of $\trop_+$ are simply denoted by elements $a \in \trop \setminus \{\zero\}$, while the elements of $\trop_-$ are denoted by $\tminus a$. 
The \emph{modulus} of $x \in \strop$ is defined as $|x| := a$ if $x = a$ or $x = \tminus a$, and $|\zero| := \zero$. Similarly, the \emph{sign} of $x \in \strop$ is defined by $\sign(x) = +1$ if $x \in \trop_+$, $\sign(x) = -1$ if $x \in \trop_-$, and $\sign(\zero) = 0$. 

While the tropical addition of signed elements may not be well defined, we can extend the multiplication over $x, y \in \strop$, by defining $x \tdot y$ as the unique element of $\strop$ with modulus $|x| \tdot |y|$ and sign $(\sign(x) \times \sign(y))$. For instance, $(\tminus 3) \tdot 4 = \tminus 7$, and $(\tminus 2) \tdot (\tminus 4) = 6$. 
The exponentiation $x^{\tdot k}$ is generalized to the case  $x \in \strop$ as well. 
For any $x \in \strop$, we use the notation $\tminus x$ as a shorthand for the operation $(\tminus \unit) \tdot x$. The \emph{positive} and \emph{negative parts} $x^+$ and $x^-$ of an element $x \in \strop$ are defined by $(x^+, x^-) := (x,\zero)$ if $x \in \trop_+$, $(\zero,\tminus x)$ if $x \in \trop_-$, and $(\zero, \zero)$ if $x = \zero$. We extend this notation to vectors and matrices entrywise.

A matrix $M \in \strop^{n \times n}$ is said to be \emph{generic} if the following maximum
\begin{equation}
\max\big\{|M_{1\sigma(1)}| \tdot \dots \tdot |M_{n\sigma(n)}| : \sigma \in S_n \big\} \label{eq:tper}
\end{equation}
is not equal to $\zero$, and is attained by only one permutation $\sigma$ in the symmetric group $S_n$. A matrix $A \in \strop^{m \times n}$ is \emph{strongly non-degenerate} if all of its square submatrices are generic. In particular, the coefficients of $A$ are not null (tropically).

\subsubsection*{Generalized Puiseux series} 

Tropical arithmetic can be intuitively illustrated by the arithmetic over asymptotic orders of magnitude. For instance, if we denote by $\Theta(t^a)$ the equivalence class of real functions of $t$ which belong to some interval $[K t^a, K' t^a]$ when $t \rightarrow +\infty$ ($0 < K \leq K'$), we have $\Theta(t^a) + \Theta(t^b) = \Theta(t^{\max(a,b)})$, and $\Theta(t^a) \times \Theta(t^b) = \Theta(t^{a+b})$. We use generalized Puiseux series as a way to manipulate such orders of magnitude in a formal way. 

A \emph{(real) generalized Puiseux series} (or \emph{Puiseux series} for short) is a formal series $\bo x$ in the indeterminate $t$ of the form $x_{\alpha_1} t^{\alpha_1} + x_{\alpha_2} t^{\alpha_2} + \dots$, where $\alpha_i \in \R$, $x_{\alpha_i} \in \R \setminus \{0\}$, and the sequence of the $\alpha_i$ is decreasing, and either finite or unbounded. The set of generalized Puiseux series forms a field that we denote $\puiseux$. Given a Puiseux series $\bo x$ as above, the largest exponent $\alpha_1$ is called the \emph{valuation} of $\bo x$, and is denoted by $\val(\bo x)$. By convention, the valuation of the null series $\bo x = 0$ is defined as $\zero = -\infty$. Thus the valuation $\val(\cdot)$ maps $\puiseux$ to $\trop$. 

A Puiseux series $\bo x$ is \emph{positive}, which is denoted by $\bo x > 0$, if the coefficient $x_{\val(x)}$ of the term with largest exponent in $\bo x$ is positive. We denote by $\leq$ the total order over $\puiseux$ defined by $\bo x \leq \bo y$ if $\bo x = \bo y$ or $\bo y - \bo x > 0$. Then, we define the \emph{signed valuation} $\sval(\bo x)$ of $\bo x$ as the element of $\strop$ given by $\val(\bo x)$ if $\bo x > 0$, $\tminus (\val(\bo x))$ if $\bo x < 0$, and $\zero$ if $\bo x = 0$. Given $x \in \strop$, we also denote by $\sval^{-1}(x)$ the set of Puiseux series $\bo x$ such that $\sval(\bo x) = x$. Valuation, signed valuation and its inverse are extended to vectors and matrices coordinate-wise. 

As discussed above, the arithmetic operations over $\puiseux$ and $\trop$ are related. For instance, for all $\bo x, \bo y \geq 0$, we have $\val(\bo x + \bo y) = \max(\val(\bo x), \val(\bo y))$. Similarly, if $\bo x, \bo y \in \puiseux$, then $\sval(\bo x \bo y) = \sval(\bo x) \tdot \sval(\bo y)$. More generally, the valuation will be used to transfer ``classical'' results to the tropical setting. In particular, convex polyhedra and linear programs over generalized Puiseux series essentially behave as over $\R$ ($\puiseux$ is a real-closed field~\cite{Markwig2007}, so Tarski's transfer principle applies). The simplex algorithm can be defined over $\puiseux$ as usual, and the valuation map will allow us to relate it with the tropical simplex algorithm.

\subsubsection*{General notations.} Given a matrix $A$ of dimension $m \times n$, and two subsets $I \subset [m]$ and $J \subset [n]$, we denote by $A_{I \times J}$ the submatrix formed by the rows and the columns of $A$ respectively indexed by $i \in I$ and $j \in J$. If $J = [n]$, we simply use the notation $A_I$. Similarly, if $i \in [n]$, $A_i$ represents the $i$-th line of $A$. The transpose matrix of $A$ is denoted by $\transpose{A}$, and the cardinality of a set $I$ by $|I|$. Given $s_1, \dots, s_n$, we denote by $\diag(s_1, \dots, s_n)$ (resp.\ $\tdiag(s_1, \dots, s_n)$) the matrix of dimension $n \times n$, with coefficients $s_i$ on the diagonal, and $0$ (resp.\ $\zero$) elsewhere.

\subsection{The tropical simplex method}\label{subsec:tropical_simplex}

The tropical simplex method, introduced in~\cite{AllamigeonBenchimolGaubertJoswig2013}, allows to solve tropical analogues of linear programming problems:
\begin{equation}
\begin{array}{r@{\quad}l}
\text{minimize} & \transpose{c} \tdot x \qquad (x \in \trop^n) \\
\text{subject to} & 
x \geq \zero , \ A^+ \tdot x \tplus b^+ \geq A^- \tdot x \tplus b^- 
\end{array} \tag*{$\tropLP(A,b,c)$} \label{eq:trop_lp}
\end{equation}
where $A \in \strop^{m \times n}$, $b \in \strop^m$, and $c \in \strop^n$. We denote by $\Pcal$ the tropical polyhedron defined by the constraints of~\ref{eq:trop_lp}. Note that the inequalities $x \geq \zero$ are trivially satisfied by any $x \in \trop^n$, hence they are superfluous. However, as we shall see, they are involved in the definition of tropical basic points, which is why we need to keep them.

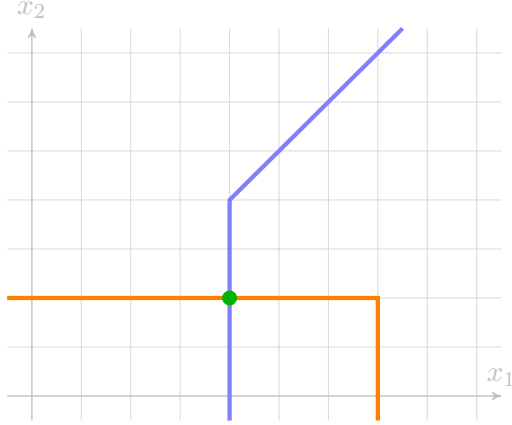
\begin{figure}
\begin{center}
\begin{tikzpicture}[scale=0.65]
\draw[grid] (-0.5,-0.5) grid (9.5,7.5);
\draw[axis] (-0.5,0) -- (9.5,0) node[color=gray!50,above] {$x_1$};
\draw[axis] (0,-0.5) -- (0,7.5) node[color=gray!50,above] {$x_2$};

\draw[orange,ultra thick] (-0.5,2) -- (7,2) -- (7,-0.5);

\draw[blue!50,ultra thick] (7.5,7.5) -- (4,4) -- (4,-0.5);

\filldraw[green!70!black] (4,2) circle (4pt);
\end{tikzpicture}
\end{center}
\caption{Intersection of two tropical hyperplanes in general position. The hyperplanes are given by the equalities $x_1 = \max(x_2, 4)$ (blue) and $\max(-5 + x_1, x_2) = 2$ (orange).}\label{fig:hyperplanes}
\end{figure}

Similarly to the classical simplex method, the principle of the tropical simplex method is to pivot over the (feasible) tropical basic points, while decreasing the objective function $x \mapsto \transpose{c} \tdot x$. It handles tropical linear programs which satisfy a certain non-degeneracy assumption. Here, we will make the following sufficient assumption:
\begin{assumption}\label{assump:non_degenerate}
The matrices $\bigl(A \  b\bigr)$ and $\bigl(\begin{smallmatrix} A \\ c \end{smallmatrix}\bigr)$ are strongly non-degenerate.
\end{assumption}
In this setting, a \emph{basis} is a couple $(I,J)$ where $I \subset [m]$, $J \subset [n]$ satisfy $|I| + |J| = n$. Note that we will often manipulate $(I,J)$ through the disjoint union $I \disjcup J$ of $I$ and $J$. By a tropical analogue of Cramer theorem~\cite{PlusCDC90}, it can be shown that, under Assumption~\ref{assump:non_degenerate}, the system 
\begin{equation}
\begin{aligned}
A^+_I \tdot x \tplus b^+_I & = A^-_I \tdot x \tplus b^-_I \\
x_J & = \zero
\end{aligned} \label{eq:basic_point}
\end{equation}
admits at most one solution in $\trop^n$. If this system admits a solution $\bar{x}$, the latter is referred to as the \emph{basic point} associated with the basis $(I,J)$. When $\bar{x}$ belongs to $\Pcal$, it is called a \emph{feasible basic point}, and $(I,J)$ is a \emph{feasible basis}. 

\begin{remark}
Every equality in the system described in~\eqref{eq:basic_point} defines a \emph{tropical hyperplane}. Assumption~\ref{assump:non_degenerate} ensures that these hyperplanes are in general position, so that the intersection of $n$ such hyperplanes is either empty, or reduced to a singleton. We refer to Figure~\ref{fig:hyperplanes} for an illustration. This provides a geometric interpretation of basic points in terms of the arrangement of the tropical hyperplanes associated with the system $x \geq \zero$, $A^+ \tdot x \tplus b^+ \geq A^- \tdot x \tplus b^-$.
\end{remark}

The execution of the tropical simplex method on~\ref{eq:trop_lp} is related with the execution of the classical simplex method on a lifting linear program over Puiseux series. More precisely, a \emph{lift} of~\ref{eq:trop_lp} is a linear program over Puiseux series of the form:
\begin{equation}
\begin{array}{r@{\quad}l}
\text{minimize} & \transpose{\bo c} \bo x \qquad (\bo x \in \puiseux^n) \\
\text{subject to} & 
\bo x \geq 0 , \ \bo A \bo x + \bo b \geq 0 
\end{array} \tag*{$\LP(\bo A,\bo b,\bo c)$} \label{eq:lp}
\end{equation}
where $\bo A \in \sval^{-1}(A)$, $\bo b \in \sval^{-1}(b)$, and $\bo c \in \sval^{-1}(c)$. Let us denote by $\Pbcal$ the convex polyhedron defined by the inequalities of~\ref{eq:lp}. Then, $\Pcal$ and $\Pbcal$ have precisely the same (feasible) bases, and the map $\bo x \mapsto \val(\bo x)$ induces a one-to-one correspondence between their basic points~\cite[Proposition~17]{AllamigeonBenchimolGaubertJoswig2013}. Besides, if $\bo x^* \in \Pbcal$ minimizes the function $\bo x \mapsto \transpose{\bo c} \bo x$, then $\val(\bo x^*) \in \Pcal$ minimizes $x \mapsto \transpose{c} \tdot x$. Note that under Assumption~\ref{assump:non_degenerate}, the linear program $\LP(\bo A,\bo b,\bo c)$ is non-degenerate, in the sense that no minor of $\bigl(\bo A \  \bo b\bigr)$ and $\bigl(\begin{smallmatrix} \bo A \\ \bo c \end{smallmatrix}\bigr)$ is null.

Moreover, both simplex methods iterate over basic points in the same way. Starting from a basic point of basis $(I,J)$, they pivot to an adjacent basic point associated with a basis $(I', J')$ such that $I' \disjcup J' = (I \disjcup J) \setminus \{\vout\} \cup \{\vin\}$, for some $\vout \in I \disjcup J$, $\vin \not \in I \disjcup J$. The element $\vout$ is called the \emph{leaving variable}, and is provided by the pivoting rule. The integer $\vin$ 
is uniquely determined when the problem is non-degenerate.\footnote{Our terminology follows the one of the dual simplex method. This inversion comes from the fact that our notion of basis actually corresponds to the set of ``non-basic'' variables in a linear program written with slack variables.}

As a consequence, under Assumption~\ref{assump:non_degenerate}, if the classical and tropical simplex methods both start from the basis $(I,J)$ and select the same leaving variable $\vout$, they pivot to the same basis $(I',J')$. In other words, the tropical simplex method traces the image under the map $\val(\cdot)$ of the path followed by the classical simplex method, provided that they use compatible pivoting rules, \ie\ at any feasible basis, both rules select the same leaving variable. 

We recall that, given a tropical basic point with basis $(I,J)$ and a leaving variable $\vout$, the operation of pivoting to the next tropical basic point can be done in time $O(n(m+n))$~\cite[Theorem~33]{AllamigeonBenchimolGaubertJoswig2013}. However, this pivoting operation is limited to basic points with no $-\infty$ coefficients. We present in Figure~\ref{fig:tropical_pivoting} a simpler, but more expensive, pivoting operation, which handles the general case. The algorithm $\Call{Pivot}{(I,J),\vout}$ consists in testing all the $m$ possibly entering variables $\vin$ in the complement of $I \disjcup J$. By the non-degeneracy assumptions, only one variable can lead to a feasible basis. Each candidate basic point can be computed in time $O(n^3)$ (see~\cite{PlusCDC90}), and the feasibility can be checked in $O(m n)$ by testing the $m$ inequalities in $A^+ \tdot x \tplus b^+ \geq A^- \tdot x \tplus b^- $. The total complexity of our pivoting operation is therefore in $O(m n(m + n^2))$. 

\begin{figure}[t]
\begin{small}
\begin{algorithmic}[1]
\Procedure {Pivot}{$(I,J), \vout$}
\ForAll{$\vin \in ([m] \setminus I) \disjcup ([n] \setminus J)$}
    \State Let $(I',J'$) be defined by $I' \disjcup J' = (I \disjcup J) \setminus \{\vout\} \cup \{\vin\}$
    \If{the system $A^+_{I'} \tdot x \tplus b^+_{I'} = A^-_{I'} \tdot x \tplus b^-_{I'}, \ x_{J'} = \zero$ admits a solution $\bar{x} \in \trop^n$, and $\bar{x} \in \Pcal$}
        \State \Return $(I',J')$
    \EndIf
\EndFor
\EndProcedure
\end{algorithmic}
\end{small}
\caption{A naive implementation of the tropical pivoting operation}\label{fig:tropical_pivoting}
\end{figure}

\begin{remark}\label{remark:edge}
The geometric interpretation of the pivoting operation in the tropical setting is analogous to the classical setting. Let $x$ and $y$ be two adjacent tropical basic points, respectively associated with the bases $(I,J)$ and $(I',J')$, so that $I' \disjcup J' = (I \disjcup J) \setminus \{\vout\} \cup \{\vin\}$ for some $\vout \in I \disjcup J$, $\vin \not \in I \disjcup J$. Let $K \subset [m]$ and $L \subset [n]$ such that $K \disjcup L = (I \disjcup J) \setminus \{ \vout \}$. Then, the set $\Ecal$ of points $z \in \Pcal$ which satisfy the equalities 
\[
\begin{aligned}
A^+_K \tdot z \tplus b^+_K & = A^-_K \tdot z \tplus b^-_K \\
z_L & = \zero
\end{aligned}
\]
is called a \emph{tropical edge} of $\Pcal$. Geometrically, it coincides with the \emph{tropical segment} between the two points $x,y$, which consists of the set $\{\lambda \tdot x \tplus \mu \tdot y : \lambda, \mu \in \trop, \ \lambda \tplus \mu = \unit\}$. As shown in~\cite[Proposition~18]{AllamigeonBenchimolGaubertJoswig2013}, this tropical edge is equal to the image under the map $\val(\cdot)$ of the edge $\Ebcal$ of the polyhedron $\Pbcal$ which connects the two basic points associated with the bases $(I,J)$ and $(I', J')$. 
\end{remark}

\begin{example}
We provide in Figure~\ref{fig:tropical_polyhedron} an example of tropical polyhedron in dimension $2$. It is defined by the following five inequalities:
\[
\begin{aligned}
\max(x,1 + y) & \geq 3 \\ 
y & \geq \max(-10 + x, 1) \\
\max(y, 4) & \geq -3 + x \\ 
8 & \geq \max(x, 2 + y) \\
4 + x & \geq \max(y, 5) 
\end{aligned}
\]
or, equivalently, by the system $x \geq \zero, \ A^+ \tdot x \tplus b^+ \geq A^- \tdot x \tplus b^-$ where:
\[
A = \begin{pmatrix} 0 & 1 \\ \tminus (-10) & 0 \\ \tminus (-3) & 0 \\ \tminus 0 & \tminus 2 \\ 4 & \tminus 0 \end{pmatrix} \qquad 
b = \begin{pmatrix} \tminus 3 \\ \tminus 1 \\ 4 \\ 8 \\ \tminus 5\end{pmatrix} \ .
\]
Note that the inequalities $x \geq \zero$ are never active (our polyhedron is bounded in $\R^2$), so that all feasible bases are of the form $(I, \emptyset)$, where $I \subset [5]$ has cardinality $2$. The basic points associated with the bases $(\{2,3\}, \emptyset)$ and $(\{3, 4\}, \emptyset)$ are depicted in blue and orange respectively. The tropical edge between them is represented in green.
\end{example}

\begin{figure}
\begin{center}
\begin{tikzpicture}[scale=0.65]
\draw[grid] (-0.5,-0.5) grid (9.5,7.5);
\draw[axis] (-0.5,0) -- (9.5,0) node[color=gray!50,above] {$x_1$};
\draw[axis] (0,-0.5) -- (0,7.5) node[color=gray!50,above] {$x_2$};

\filldraw[convex] (1,2) -- (3,2) -- (3,1) -- (7,1) -- (7,4) -- (8,5) -- (8,6) -- (2,6) -- (1,5) -- cycle;
\draw[convexborder] (1,2) -- (3,2) -- (3,1) -- (7,1) -- (7,4) -- (8,5) -- (8,6) -- (2,6) -- (1,5) -- cycle;

\draw[green!70!black, line width=1.5pt] (7,1) -- (7,4) -- (8,5);
\filldraw[point] (7,1) circle (4pt);
\filldraw[orange] (8,5) circle (4pt);

\end{tikzpicture}
\end{center}
\caption{A tropical polyhedron (in gray), two basic points (in blue and orange) and a tropical edge (in green).}\label{fig:tropical_polyhedron}
\end{figure}
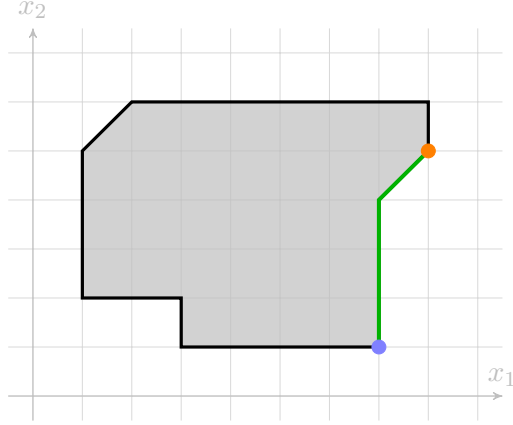

\section{Tropicalizing the shadow-vertex simplex algorithm}\label{sec:shadow_vertex}

\subsection{The classical shadow-vertex pivoting rule}

Given $\bo u, \bo v \in \puiseux^{n}$, the shadow-vertex rule aims at solving the following parametric family of linear programs for increasing values of $\bo \lambda \geq 0$:
\begin{equation}
\begin{array}{r@{\quad}l}
\text{minimize} & (\transpose{\bo u} - \bo \lambda \transpose{\bo v}) \bo x \qquad (\bo x \in \puiseux^n) \\
\text{subject to} & 
\bo x \geq 0 , \ \bo A \bo x + \bo b \geq 0 
\end{array}  \label{eq:parametric_lp}
\end{equation}
The vectors $\bo u$ and $\bo v$ are respectively called \emph{objective} and \emph{co-objective} vectors. The input of~\eqref{eq:parametric_lp} is supposed to satisfy a genericity property. Here, we will assume that no minor of $\bigl(\bo A \ \bo b \bigr)$ and $\bigl(\transpose{\bo A} \ \bo u \ \bo v\bigr)$ is null. 

When $\bo \lambda$ is continuously increased from $0$, the basic points of $\Pbcal$ minimizing the function $\bo x \mapsto (\transpose{\bo u} - \bo \lambda \transpose{\bo v}) \bo x$ form a sequence $\bar{\bo x}^0, \dots, \bar{\bo x}^p$. 
The shadow-vertex rule amounts to iterate over this sequence. It relies on the reduced cost vectors w.r.t.\ the objective and co-objective vectors. 
Given a basis $(I,J)$, the \emph{reduced cost vector} $\bo y^{(I,J)} \in \puiseux^{I \disjcup J}$ w.r.t.\ the objective vector $\bo u$ is defined as the unique solution $\bo y$ of the system $\transpose{(\bo A'_{I \disjcup J})} \bo y = \bo u$, where $\bo A' = \begin{psmallmatrix} \boid \\ \bo A\end{psmallmatrix}$, and $\boid$ is the identity matrix.
The reduced cost vector $\bo z^{(I,J)}$ w.r.t.\ the co-objective vector $\bo v$ can be defined similarly, by replacing $\bo u$ by $\bo v$. 
Then, at basis $(I,J)$, the shadow-vertex rule selects the leaving variable $\vout \in I \disjcup J$ such that $\bo y^{(I,J)}_{\vout} > 0$, $\bo z^{(I,J)}_{\vout} > 0$, and:
\begin{equation}
\bo y^{(I,J)}_{\vout} / \bo z^{(I,J)}_{\vout} =
\min 
\bigl\{ 
\bo y^{(I,J)}_l / \bo z^{(I,J)}_l : l \in I \disjcup J, \ \bo y^{(I,J)}_l > 0 \ \text{and} \ \bo z^{(I,J)}_l > 0
\bigr\} \ . \label{eq:shadow_vertex}
\end{equation}
Note that $\vout$ is unique under the non-degeneracy assumptions. We refer to~\cite[Chapter~1.3]{Borgwardt1987} for a proof of~\eqref{eq:shadow_vertex}. 
In the following, we will denote by $\sigmags(\bo A, \bo u, \bo v)$ the function which, given a basis $(I,J)$, returns the leaving variable provided by the shadow-vertex rule with objective and co-objective vectors $\bo u$ and $\bo v$. If there is no $l \in I \disjcup J$ such that $\bo y^{(I,J)}_l > 0$ and $\bo z^{(I,J)}_l > 0$, the function $\sigmags(\bo A, \bo u, \bo v)(I,J)$ will be supposed to return the  special value $\None$. It can be shown that this happens if, and only if, the basic point associated with the basis $(I,J)$ maximizes the co-objective function $\bo x \mapsto \transpose{\bo v} \bo x$ (see~\cite[Chapter~1.3]{Borgwardt1987}).

\subsection{Semi-algebraic pivoting rules and their tropical counterpart}

In this section, we focus on the problem of finding a tropical pivoting rule $\sigmagstrop$ compatible with the shadow-vertex rule $\sigmags$. More formally, we aim at defining a function $\sigmagstrop(A,u,v)$ parametrized by a tropical matrix $A \in \strop^{m \times n}$, and objective and co-objective vectors $u, v \in \strop^{n}$, such that:
\[
\sigmagstrop(A, u, v)(I,J) = \sigmags(\bo A, \bo u, \bo v)(I,J) \qquad \text{for any basis} \ (I,J) \ ,
\]
for all $\bo A \in \sval^{-1}(A)$, $\bo u \in \sval^{-1}(u)$, and $\bo v \in \sval^{-1}(v)$. In this case, $\sigmagstrop$ will be said to be \emph{compatible with} $\sigmags$ \emph{on the instance} $(A,u,v)$.

\subsubsection*{Tropical polynomials.}

The connection we use between the classical and tropical worlds relies on polynomials over generalized Puiseux series. 

Let $P \in \puiseux[X_1, \dots, X_p]$ be a multivariate polynomial, and suppose that it is of the form $\sum_{\alpha \in S} \bo c_\alpha X_1^{\alpha_1} \dots X_p^{\alpha_p}$, where $S \subset \mathbb{N}^{p}$, and $\bo c_\alpha \in \puiseux \setminus \{0\}$ for all $\alpha \in S$. The set $S$ is called the \emph{support} of $P$. We associate a tropical polynomial $\tropPol(P) \in \strop[][X_1, \dots, X_p]$ defined as the following formal expression:
\[
\tropPol(P) := \tsum_{\alpha \in S} c_\alpha \tdot X_1^{\tdot \alpha_1} \tdot \dots \tdot X_p^{\tdot \alpha_p} \ ,
\]
with $c_\alpha := \sval(\bo c_\alpha)$ for all $\alpha \in S$. Given $x \in \strop^{p}$, we say that the polynomial $\tropPol(P)$ \emph{vanishes} on $x$ if the following maximum
\begin{equation}
\max\bigl\{ |c_\alpha| \tdot |x_1|^{\tdot \alpha_1} \tdot \dots \tdot |x_p|^{\tdot \alpha_p} : \alpha \in S \bigr\} \label{eq:poly_modulus}
\end{equation}
is reached at least twice, or is equal to $\zero$. If $\tropPol(P)$ does not vanish on $x$, we define $\tropPol(P)(x) \in \strop$ as follows:
\[
\tropPol(P)(x)  := c_{\alpha^*} \tdot x_1^{\tdot \alpha^*_1} \tdot \dots \tdot x_p^{\tdot \alpha^*_p} \ ,
\]
where $\alpha^*$ is the unique element of $S$ reaching the maximum in~\eqref{eq:poly_modulus}. 
The following lemma relates the values of $P$ and $\tropPol(P)$: 
\begin{lemma}\label{lemma:poly_sign}
Let $x \in \strop^{p}$, and suppose that $\tropPol(P)$ does not vanish on $x$. Then, for any $\bo x \in \sval^{-1}(x)$, we have $\sval(P(\bo x)) = \tropPol(P)(x)$. In particular, the sign of $P(\bo x)$ is equal to the sign of $\tropPol(P)(x)$.
\end{lemma}

\begin{proof}
Let $\bo x \in \sval^{-1}(x)$. Given $\alpha \in S$, the valuation of the term $\bo c_\alpha \bo x_1^{\alpha_1} \dots \bo x_p^{\alpha_p}$ is equal to $|c_\alpha| \tdot |x_1|^{\tdot \alpha_1} \tdot \dots \tdot |x_p|^{\tdot \alpha_p}$. Similarly, its signed valuation is given by $c_\alpha \tdot x_1^{\tdot \alpha_1} \tdot \dots \tdot x_p^{\tdot \alpha_p}$. As the maximum in~\eqref{eq:poly_modulus} is reached by only one element $\alpha^* \in S$, we deduce that $\bo c_{\alpha^*} \bo x_1^{\alpha^*_1} \dots \bo x_p^{\alpha^*_p}$ is the unique monomial with largest valuation in $P(\bo x)$. Thus, the signed valuation of $P(\bo x)$ coincides with the signed valuation of $\bo c_{\alpha^*} \bo x_1^{\alpha^*_1} \dots \bo x_p^{\alpha^*_p}$, and is equal to $\tropPol(P)(x) = c_{\alpha^*} \tdot x_1^{\tdot \alpha^*_1} \tdot \dots \tdot x_p^{\tdot \alpha^*_p}$.
\end{proof}

Following this, we can introduce determinants of tropical matrices. Let us define 
\[
\tdet_n(X) := \tsum_{\sigma \in S_n} \tsign(\sigma) \tdot X_{1 \sigma(1)} \tdot \dots \tdot X_{n\sigma(n)} \ ,
\]
where $\tsign(\sigma) := \unit$ if the permutation $\sigma$ is even, $\tminus \unit$ otherwise. The polynomial $\tdet_n$ is simply denoted $\tdet$ when there is no ambiguity. If $M \in \strop^{n \times n}$, the \emph{tropical determinant} of $M$ is defined as $\tdet(M)$ when the polynomial $\tdet$ does not vanish on $M$. Note that the latter condition is equivalent to the fact that $M$ is generic. In this case, $\tdet(M)$ can be computed in time complexity $O(n^3)$, by solving an assignment problem over the bipartite graph with node set $[n] \disjcup [n]$, in which every arc $(i,j)$ is equipped with the weight $|M_{ij}|$. Indeed, the maximum weight matching provides the unique permutation $\sigma^* \in S_n$ reaching the maximum in~\eqref{eq:tper}, and by definition, $\tdet(M)$ is given by $\tsign(\sigma^*) \tdot M_{1 \sigma^*(1)} \tdot \dots \tdot M_{n \sigma^*(n)}$. 

\begin{example}
The determinant of $(2 \times 2)$-matrices is given by the polynomial ${\det}_2 = X_{1,1} X_{2,2} - X_{1,2} X_{2,1}$, and the corresponding tropical polynomial is $\tdet_2 = (X_{1,1} \tdot X_{2,2}) \tplus (\tminus (X_{1,2} \tdot X_{2,1}))$. Let us consider the tropical matrix 
$M := 
\begin{psmallmatrix}
3 & \tminus 2 \\
1 & \tminus 1 
\end{psmallmatrix}$, and let $\bo M := \begin{psmallmatrix}
2 t^{3} + \cdots & -t^{2} + \cdots \\
4 t + \cdots & -9 t + \cdots  
\end{psmallmatrix}$ be an arbitrary lift of $M$ (the dots represent terms of the series which have a smaller exponent, which we left unspecified). The tropical polynomial $\tdet_2$ does not vanish on $M$, since we have:
\begin{equation}
\max(|M_{1,1}| \tdot |M_{2,2}|, |M_{1,2}| \tdot |M_{2,1}|) = \max(3 \tdot 1, 1 \tdot 2) = \max(4,3) \ . \label{eq:example}
\end{equation}
Hence, the term reaching the maximum in~\eqref{eq:example} is associated with the monomial $X_{1,1} \tdot X_{2,2}$ of $\tdet_2$, so that $\tdet_2(M) = M_{1,1} \tdot M_{2,2} = \tminus 4$. On the other hand, the determinant of $\bo M$ is of the form $-18 t^{4} + \cdots + 4 t^{3} + \cdots$. Consequently, we indeed have $\sval(\det_2(\bo M)) = \tdet_2(M)$, as expected.

Moreover, the term $|M_{1,1}| \tdot |M_{2,2}|$ attaining the maximum in~\eqref{eq:example} is given by the maximum weight assignment $(1,1), (2,2)$ in the following bipartite graph with weights $|M_{ij}|$:

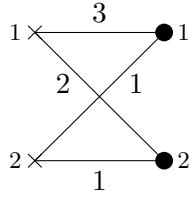
\begin{center}
\begin{tikzpicture}[scale=0.85,row/.style={draw,cross out,inner sep=0.5ex},col/.style={fill,circle,inner sep=0.5ex}]
\path node[row] (h1) at (0,0) {} (h1) node[left=0.25ex,font=\footnotesize] {$1$};
\path node[col] (j1) at (2,0) {} (j1) node[right=0.25ex,font=\footnotesize] {$1$};
\path node[row] (h2) at (0,-2) {} (h2) node[left=0.25ex,font=\footnotesize] {$2$};
\path node[col] (j2) at (2,-2) {} (j2) node[right=0.25ex,font=\footnotesize] {$2$};

\draw (h1.center) -- node[above] {$3$} (j1.center) -- node[above right=-0.5ex and 1.5ex] {$1$} (h2.center) -- node[below] {$1$} (j2.center) -- node[above left=-0.5ex and 1.5ex] {$2$} (h1.center);
\end{tikzpicture}
\end{center}
\end{example}

\subsubsection*{The shadow-vertex rule as a semi-algebraic rule.} We claim that the shadow-vertex rule is a \emph{semi-algebraic rule}, in the sense that the leaving variable returned by $\sigmags(\bo A, \bo u, \bo v)(I,J)$ only depends on the current basis $(I,J)$ and on the signs of finitely many polynomials taken on the matrix $\bo M := \begin{psmallmatrix}\bo A \\ \transpose{\bo u} \\ \transpose{\bo v}\end{psmallmatrix}$. To make the notations simpler, we fix a basis $(I,J)$, and we respectively denote by $\bo y$ and $\bo z$ the reduced cost vectors $\bo y^{(I,J)}$ and $\bo z^{(I,J)}$. We also define $\compl{J} := [n] \setminus J$.

Let us denote by $P_{K \times L}$ the polynomial given by the $(K \times L)$-minor of the matrix $X = (X_{ij})$ of formal variables, for any $K \subset [m+2]$ and $L \subset [n]$ such that $|K| = |L|$. For instance, if $K = \{1,2\}$ and $L = \{3,4\}$, $P_{K \times L}$ is given by the determinant of the submatrix $\begin{psmallmatrix} 
X_{1,3} & X_{1,4} \\
X_{2,3} & X_{2,4}
\end{psmallmatrix}$, \ie~$P_{K \times L} = X_{1,3} X_{2,4} - X_{2,3} X_{1,4}$.  For all $l \in I \disjcup J$, we define two polynomials $Q_l$ and $R_l$ as follows:  
\begin{align*}
Q_i & := P_{(I \setminus \{i\} \cup \{m+1\}) \times \compl{J}}  & 
R_i & := P_{(I \setminus \{i\} \cup \{m+2\}) \times \compl{J}} && \text{when} \ i \in I \ , \\
Q_j & := P_{(I \cup \{m+1\}) \times (\compl{J} \cup \{j\})}  &
R_j & := P_{(I \cup \{m+2\}) \times (\compl{J} \cup \{j\})} 
&& \text{when} \ j \in J \  .
\end{align*}

\begin{lemma}\label{lemma:reduced_cost}
For all $l \in I \disjcup J$, 
\[
\bo y_l = s_l \, Q_l(\bo M) / P_{I \times \compl{J}}(\bo M) \ , \qquad
\bo z_l = s_l \, R_l(\bo M) / P_{I \times \compl{J}}(\bo M) \ ,
\]
where $s_l$ is a constant in $\{\pm 1\}$ which only depends on the integer $l$ and the sets $I$ and $J$.
\end{lemma}

\begin{proof}
We restrict our attention to the vector $\bo y$ (the proof is similar for the vector $\bo z$). Recall that $\bo y$ is given by the following system:
\begin{align*}
\bo y_J + \transpose{(\bo A_{I \times J})} \bo y_I & = \bo u_J \\
\transpose{(\bo A_{I \times \compl{J}})} \bo y_I & = \bo u_{\compl{J}} 
\end{align*}
By the latter part, for all $i \in I$, 
\[
\bo y_i = (-1)^{n-\idx(i)} \det\begin{pmatrix} \bo A_{(I \setminus \{i\}) \times \compl{J}} \\ \transpose{(\bo u_{\compl{J}})}\end{pmatrix} / \det(\bo A_{I \times \compl{J}}) \ ,
\]
where $\idx(i)$ represents the index of $i$ in the ordered set $I$. It follows that for all $j \in J$, we have:
\[
\bo y_j \det(\bo A_{I \times \compl{J}}) = \bo u_j \det(\bo A_{I \times \compl{J}}) - \sum_{i \in I} (-1)^{n-\idx(i)} \bo A_{ij} \det \begin{pmatrix}
\bo A_{(I \setminus \{i\}) \times \compl{J}} \\
\transpose{(\bo u_{\compl{J}})}
\end{pmatrix} \enspace.
\]
By developing the determinant of $\begin{psmallmatrix} \bo A_{I \times (\compl{J} \cup \{j\})} \\ \transpose{(\bo u_{\compl{J} \cup \{j\}})} \end{psmallmatrix}$ w.r.t.~ the column $\begin{psmallmatrix} \bo A_{I \times \{j\}} \\ \bo u_j \end{psmallmatrix}$, we obtain that:
\[
\bo y_j \det(\bo A_{I \times \compl{J}}) = 
(-1)^{n + 1 - \idx'(j)} \det \begin{pmatrix} \bo A_{I \times (\compl{J} \cup \{j\})} \\ \transpose{(\bo u_{\compl{J} \cup \{j\}})} \end{pmatrix} \ , 
\]
where $\idx'(j)$ stands for the index of $j$ in the ordered set $\compl{J} \cup \{j\}$. 
\end{proof}

As a consequence of Lemma~\ref{lemma:reduced_cost}, the properties $\bo y_l > 0$, $\bo z_l > 0$ can be tested by determining the signs of $Q_l(\bo M)$, $R_l(\bo M)$ and $P_{I \times \compl{J}}(\bo M)$. Moreover, we have:
\[
\bo y_l/\bo z_l = Q_l(\bo M)/R_l(\bo M) \ . 
\] 
Hence, the comparison of two ratios $\bo y_k / \bo z_k$ and $\bo y_l / \bo z_l$ involved in the shadow-vertex rule can be made by evaluating the sign of a polynomial of the form $T_{kl} := Q_k R_l - Q_l R_k$ on the matrix $\bo M$. This shows that the shadow-vertex rule is semi-algebraic.

\subsubsection*{Tropical shadow-vertex rule.}

Following the previous discussion, we can express $\sigmags(\bo A, \bo u, \bo v)$ as a function defined in terms of the signs of some minors $\det(\bo M_{K \times L})$, and of the signs of the $T_{kl}(\bo M)$ ($k \neq l$). Given tropical entries $(A,u,v)$, we simply define $\sigmagstrop(A, u, v)$ as the same function, in which the minors of $\bo M$ have been substituted by the corresponding tropical minors of the matrix $M := \begin{psmallmatrix} A \\ \transpose{u} \\ \transpose{v} \end{psmallmatrix}$, and the $T_{kl}(\bo M)$ have been replaced by $\tropPol(T_{kl})(M)$. 

In more detail, the function $\sigmagstrop(A,u,v)(I,J)$ returns the unique element $\vout \in \Lambda$ such that 
\begin{equation}
\sign(\tropPol(T_{\vout l}(M)) = -s_{\vout} s_l \qquad \text{for all} \ l \in \Lambda \setminus \{\vout\} \ , \label{eq:comparison}
\end{equation}
where $\Lambda$ is the set of the elements $l \in I \disjcup J$ such that $\sign(\tropPol(Q_l)(M)) = \sign(\tropPol(R_l)(M)) = s_l \sign(\tdet(M_{I \times \compl{J}}))$. The latter condition is the tropical counterpart of the conditions $\bo y_l, \bo z_l > 0$ in the definition of $\sigmags$. Equation~\eqref{eq:comparison} is the analogue of $\bo y_{\vout}/ \bo z_{\vout} < \bo y_l/ \bo z_l$ for all $l \in \Lambda$, $l \neq \vout$. If the set $\Lambda$ is empty, we set $\sigmagstrop(A,u,v)(I,J)$ to the special value $\None$.

The main result of this section is the following:
\begin{theorem}\label{th:gs_trop}
If the matrix $\begin{psmallmatrix} A \\ \transpose{u} \\ \transpose{v} \end{psmallmatrix}$ is strongly non-degenerate, then $\sigmagstrop$ is compatible with $\sigmags$ on the instance $(A, u, v)$. 

Moreover, for all bases $(I,J)$, the leaving variable returned by $\sigmagstrop(A,u,v)(I,J)$ can be computed in time $O(n^4)$.
\end{theorem}

\begin{proof}
Let $\bo A \in \sval^{-1}(A)$, $\bo u \in \sval^{-1}(u)$, and $\bo v \in \sval^{-1}(v)$. By assumption, the matrix $M = \begin{psmallmatrix} A \\ \transpose{u} \\ \transpose{v} \end{psmallmatrix}$ is strongly non-degenerate, so that the sign of every tropical minor $\tdet(M_{K \times L})$ coincides with the sign of the corresponding minor of $\bo M := \begin{psmallmatrix}\bo A \\ \transpose{\bo u} \\ \transpose{\bo v}\end{psmallmatrix}$ by Lemma~\ref{lemma:poly_sign}. Consequently, by Lemma~\ref{lemma:reduced_cost}, the set $\Lambda$ precisely consists of the elements $l \in I \disjcup J$ such that $\bo y_l > 0$ and $\bo z_l > 0$. As discussed earlier, each tropical minor can be computed in time $O(n^3)$. Hence, the set $\Lambda$ can be determined in time $O(n^4)$. It now remains to examine the case of the polynomials $\tropPol(T_{kl})$, and to show in particular that they do not vanish on $M$. For the sake of brevity, we restrict to the case $k,l \in I$. The general case can be handled in a similar way.

Let us write the polynomial $\tropPol(T_{kl})$ under the form $\tsum_{\alpha \in S} c_\alpha \tdot X^{\tdot \alpha}$, where $S$ is the support of $T_{kl}$ (we use the notation $X^{\tdot \alpha}$ as a shorthand of $\tprod_{ij} X_{ij}^{\tdot \alpha_{ij}}$). By definition, $\tropPol(T_{kl})$ does not vanish on $M$ if, and only if, there exists a unique solution to the following maximization problem:
\begin{equation}
\begin{array}{r@{\quad}l}
\text{maximize} & |c_\alpha | \tdot |M|^{\tdot \alpha} \\
\text{subject to} & \alpha \in S 
\end{array}
\label{eq:vanish}
\end{equation}
Observe that the coefficients in $T_{kl}$ are integers. Hence, as elements of the field $\puiseux$, they are constant Puiseux series, with valuation $0 = \unit$. Besides, $M_{ij} \neq \zero$ for all $(i,j)$, thanks to the strong non-degeneracy of $M$. Then, we can simply rewrite $|c_\alpha| \tdot |M|^{\tdot \alpha} = \sum_{ij} |M_{ij}| \alpha_{ij}$ for all $\alpha \in S$. As a consequence, Problem~\eqref{eq:vanish} can be solved by considering the following classical linear program:
\begin{equation}
\begin{array}{r@{\quad}l}
\text{maximize} & \sum_{ij} |M_{ij}| \alpha_{ij} \\
\text{subject to} & \alpha \in \new(T_{kl}) 
\end{array}
\label{eq:lp_over_newton}
\end{equation}
where $\new(T_{kl}) \subset \R^{(m+2) \times n}$ is the \emph{Newton polytope} of the polynomial $T_{kl}$, defined as the convex hull of its support $S$. Since the set of vertices of the polytope $\new(T_{kl})$ is a subset of $S$, it is easy to show that $\tropPol(T_{kl})$ does not vanish on $M$ if, and only if, Problem~\eqref{eq:lp_over_newton} admits a unique solution $\alpha^*$. In this case, we have $\alpha^* \in S$, and the sign of $\tropPol(T_{kl})(M)$ is immediately given by the sign of the term $c_{\alpha^*} \tdot M^{\tdot \alpha^*}$. 

It now remains to check that Problem~\eqref{eq:lp_over_newton} indeed has a unique solution, and that the latter can be found efficiently. To this aim, we use Pl{\"u}cker quadratic relations (see for instance~\cite[Chapter 3]{GelfandKapranovZelevinsky}), which provide the identity 
$T_{kl} = P_{I \times \compl{J}} \, P_{(I \setminus \{k,l\} \cup \{m+1, m+2\}) \times \compl{J}}$. This implies that the Newton polytope of $T_{kl}$ consists in the Minkowski sum of the two polytopes $\Delta_1 := \new(P_{I \times \compl{J}})$ and $\Delta_2 := \new(P_{(I \setminus \{k,l\} \cup \{m+1, m+2\}) \times \compl{J}})$. As a result, Problem~\eqref{eq:lp_over_newton} can be decomposed into the following two linear programs:
\begin{equation}
\begin{array}{r@{\quad}l}
\text{maximize} & \sum_{ij} |M_{ij}| \alpha_{ij} \\
\text{subject to} & \alpha \in \new(\Delta_i) 
\end{array}
\qquad \quad \text{for} \ i \in \{1, 2\} \ .\label{eq:sub_problems}
\end{equation}
More precisely, the set of optimal solutions of Problem~\eqref{eq:lp_over_newton} is precisely the Minkowski sum of the set of optimal solutions of the two problems given in~\eqref{eq:sub_problems}. The polytopes $\new(\Delta_i)$ are bipartite perfect matching polytopes. Consequently, the two problems in~\eqref{eq:sub_problems} correspond to optimal assignment problems, and can be solved in $O(n^3)$. Besides, they both admit a unique solution thanks to the genericity condition on $M$. 

To summarize, $\tropPol(T_{kl})$ does not vanish on $M$, and the sign of $\tropPol(T_{kl})(M)$ can be computed in time $O(n^3)$. By Lemmas~\ref{lemma:poly_sign} and~\ref{lemma:reduced_cost}, $\sign(\tropPol(T_{k l}(M)) = -s_k s_l$ if, and only if, $\bo y_k/\bo z_k < \bo y_l/ \bo z_l$. This completes the analysis of the polynomial $\tropPol(T_{kl})$. 

We deduce that $\sigmagstrop$ and $\sigmags$ are compatible. The output $\vout$ of $\sigmagstrop(A,u,v)(I,J)$ can be computed by determining the smallest element of the set $\Lambda$ according to the abstract ordering relation $\prec$ defined by $k \prec l \Longleftrightarrow \sign(\tropPol(T_{kl})(M)) = -s_k s_l$. Every comparison has time complexity $O(n^3)$, and so the result can be obtained in time $O(n^4)$.  
\end{proof}

\section{Average-case complexity of mean payoff games}\label{sec:pcbc}

\subsection{Tropicalization of the Parametric Constraint-by-Constraint algorithm}\label{subsec:classical_pcbc}

\begin{figure}[t]
\begin{small}
\begin{algorithmic}[1]
\Procedure {PCBC}{$\bo A, \bo b$}
\State $\bo u \gets \transpose{(\bo \eps, \bo \eps^2, \dots, \bo \eps^n)}$ \Comment{$0 < \bo \eps \ll 1$}
\State $\bar{\bo x} \gets \transpose{(0, \dots, 0)}$ \label{line:init}
\For{$k = 1$ to $m$} \label{line:bloop}
\If{$\bo A_k \bar{\bo x} + \bo b_k < 0$}
    \State Starting from $\bar{\bo x}$, iterate over the basic points and edges of $\Pbcal^{(k-1)}$ using the rule $\sigmags(\bo A_{[k-1]}, \bo u, \transpose{(\bo A_k)})$, until finding a point $\bar{\bo x}'$ such that $\bo A_k \bar{\bo x}' + \bo b_k = 0$. \label{line:simplex}
    \sIf{there is no such point $\bar{\bo x}'$} \Return ``Empty'' \label{line:stop}
    \sElse $\bar{\bo x} \gets \bar{\bo x}'$
\EndIf
\EndFor \label{line:eloop}
\State \Return ``Non-empty''
\EndProcedure
\end{algorithmic}
\end{small}
\caption{Parametric Constraint-by-Constraint algorithm}\label{fig:pcbc}
\end{figure}

The ave\-rage-case analysis of~\cite{AdlerKarpShamir87} applies to the so-called \emph{Parametric Constraint-by-Constraint} ($\PCBC$) algorithm. We first recall the principle of this algorithm. We restrict the presentation to polyhedral feasibility problems, following our motivation to their tropical counterparts and mean payoff games. 

The $\PCBC$ algorithm is given in Figure~\ref{fig:pcbc}. Given $\bo A \in \R^{m \times n}$, $\bo b \in \R^m$, and $k \in \{0, \dots, m\}$, we denote by $\Pbcal^{(k)}$ the polyhedron defined by the first $n+k$ inequalities of the system $\bo x \geq 0, \ \bo A \bo x + \bo b \geq 0$. The $\PCBC$ algorithm consists in determining by induction on $k = 1, \dots, m$ whether the polyhedron $\Pbcal^{(k)}$ is empty. The invariant of the loop from Lines~\lineref{line:bloop} to~\lineref{line:eloop} is that $\bar{\bo x}$ is the (unique) basic point of $\Pbcal^{(k-1)}$ minimizing the function $\bo x \mapsto \transpose{\bo u} \bo x$, where $\bo u$ is an objective vector fixed throughout the whole execution of $\PCBC$. At the $k$-th iteration of the loop, when $\bar{\bo x}$ does not satisfy the constraint $\bo A_k \bo x + \bo b_k \geq 0$, the simplex algorithm equipped with the shadow-vertex pivoting rule is used. The co-objective vector is set to $\transpose{(\bo A_k)}$. The simplex algorithm thus follows a path in $\Pbcal^{(k-1)}$ consisting of basic points and the edges between them. We stop it as soon as it discovers a point $\bar{\bo x}' \in \Pbcal^{(k-1)}$ such that $\bo A_k \bar{\bo x}' + \bo b_k = 0$ on the path. This point is obviously a basic point of $\Pbcal^{(k)}$. It follows from the definition of the shadow-vertex rule that $\bar{\bo x}'$ minimizes the objective function $\bo x \mapsto \transpose{\bo u} \bo x$ over $\Pbcal^{(k)}$. Then, $\bar{\bo x}'$ can be used as a starting point for the execution of the simplex algorithm during the $(k+1)$-th iteration. If no such point $\bar{\bo x}'$ is discovered, the simplex algorithm stops at a basic point $\bo x''$ associated with a basis $(I'',J'')$ such that $\sigmags(\bo A_{[k-1]}, \bo u, \transpose{(\bo A_k)})(I'', J'') = \None$.\footnote{\label{footnote:ray}As noted in~\cite[Section~4]{AdlerKarpShamir86}, if the simplex algorithm encounters a ray of the polyhedron of $\Pbcal^{(k-1)}$ on the path, it necessarily finds a point $\bar{\bo x}' \in \Pbcal^{(k-1)}$ such that $\bo A_k \bar{\bo x}' + \bo b_k = 0$.} In this case, $\bo x''$ maximizes the function $\bo x \mapsto \bo A_k \bo x$ over $\Pbcal^{(k-1)}$, which shows that $\bo A_k \bo x + \bo b_k < 0$ for all $\bo x \in \Pbcal^{(k-1)}$. Then, the algorithm returns ``Empty''.

The objective vector $\bo u$ is set to $\transpose{(\bo \eps, \bo \eps^2, \dots, \bo \eps^n)}$, where $\bo \eps > 0$ is a sufficiently small scalar. Since $\bo u_j > 0$ for all $j \in [n]$, the vector $\transpose{(0, \dots, 0)}$ is a basic point of $\Pbcal^{(0)} = (\R_+)^n$ minimizing $\bo x \mapsto \transpose{\bo u} \bo x$.

The $\PCBC$ algorithm is still correct when applied on inputs $\bm A, \bm b$ with entries in~$\puiseux$. This suggests to tropicalize it by using the tropical simplex algorithm equipped with the pivoting rule developed in Section~\ref{sec:shadow_vertex}. This is the purpose of the algorithm $\tropPCBC$ given in Figure~\ref{fig:trop_pcbc}. Its principle is analogous to $\PCBC$. It manipulates the sequence of tropical polyhedra $\Pcal^{(k)}$ ($0 \leq k \leq m$), which are respectively defined by the first $n+k$ inequalities of the system $x \geq \zero$, $A^+ \tdot x \tplus b^+ \geq A^- \tdot x \tplus b^-$. It also involves an objective vector of the form $u \gets \transpose{(\eps, \eps^{\tdot 2}, \dots, \eps^{\tdot n})}$, with $\eps < 0$. 

Let us describe in more detail the operations performed at Line~\lineref{line:trop_simplex}. For each visited basic point $x$ of $\Pcal^{(k-1)}$, the tropical rule $\sigmagstrop(A_{[k-1]}, u, \transpose{(A_k)})$ determines a variable $\vout$ leaving the basis $(I,J)$ associated with $x$. The tropical simplex algorithm then pivots along the edge $\Ecal$ formed by the points $z \in \Pcal^{(k-1)}$ which activate all the inequalities indexed by $l \in (I \disjcup J) \setminus \{\vout\}$ in the system defining $\Pcal^{(k-1)}$ (see Remark~\ref{remark:edge}). There exists a point $\bar{x}' \in \Ecal$ such that $A^+_k \tdot \bar{x}' \tplus b^+_k = A^-_k \tdot \bar{x}' \tplus b^-_k$ if, and only if, the pair $(I',J')$ given by $I' \disjcup J' = (I \disjcup J) \setminus \{\vout\} \cup \{k\}$ is a feasible basis of $\Pcal^{(k)}$. Indeed, such an $\bar{x}'$ is precisely characterized as the basic point of $\Pcal^{(k)}$ of basis $(I',J')$. Thus, its existence can be checked in time $O(n (m+n^2))$, as explained in Section~\ref{subsec:tropical_simplex}. If there is no such $\bar{x}'$ in $\Ecal$, we use the algorithm $\Call{Pivot}{(I,J),\vout}$ to compute the next basis of $\Pcal^{(k-1)}$. 

Note that the condition at Line~\lineref{line:stop} is satisfied when there is no basic point of $\Pcal^{(k-1)}$ to be visited anymore, \ie~when the tropical pivoting rule $\sigmagstrop(A_{[k-1]}, u, \transpose{(A_k)})$ returns $\None$. 

\begin{figure}[t]
\begin{small}
\begin{algorithmic}[1]
\Procedure{TropPCBC}{$A, b$}
\State $u \gets \transpose{(\eps, \eps^{\tdot 2}, \dots, \eps^{\tdot n})}$  \Comment{$-\infty < \eps \ll 0$}
\State $\bar{x} \gets \transpose{(\zero, \dots, \zero)}$
\For{$k = 1$ to $m$} 
\If{$A^+_k \tdot \bar{x} \tplus b^+_k < A^-_k \tdot \bar{x} \tplus b^-_k$}
    \State Starting from $\bar{x}$, iterate over the tropical basic points and edges of $\Pcal^{(k-1)}$ using the tropical rule $\sigmagstrop(A_{[k-1]}, u, \transpose{(A_k)})$ until finding a point $\bar{x}' \in \Pcal^{(k-1)}$ such that $A^+_k \tdot \bar{x}' \tplus b^+_k = A^-_k \tdot \bar{x}' \tplus b^-_k$. \label{line:trop_simplex}
    \sIf{there is no such point $\bar{x}'$} \Return ``Empty''
    \sElse $\bar{x} \gets \bar{x}'$
\EndIf
\EndFor 
\State \Return ``Non-empty''
\EndProcedure
\end{algorithmic}
\end{small}
\vskip-2ex
\caption{Tropicalization of the PCBC algorithm}\label{fig:trop_pcbc}
\end{figure}

In order to use the tropical shadow-vertex rule on the instances $(A_{[k-1]}, u, \transpose{(A_k)})$, we verify that the matrix $\begin{psmallmatrix} A \\ \transpose{u} \end{psmallmatrix}$ is strongly non-degenerate. The following lemma shows that this property holds assuming that $\eps$ is small enough:
\begin{lemma}\label{lemma:eps}
Suppose that the matrix $A$ is strongly non-degenerate, and $\eps < n (\min_{ij} |A_{ij}| - \max_{ij} |A_{ij}|)$. Then, the matrix $\begin{psmallmatrix} A \\ \transpose{u} \end{psmallmatrix}$ is strongly non-degenerate.
\end{lemma}

\begin{proof}
Let $M \in \strop^{K \times L}$ be a square submatrix of $\begin{psmallmatrix} A \\ \transpose{u} \end{psmallmatrix}$. If $M$ is a submatrix of $A$, then it is clear that $M$ is generic. Now, suppose that $M$ involves the last line $\transpose{u}$ (\ie~$m+1 \in K$), and that the maximum
\[
\max \Big\{ \bigodot_{k \in K} |M_{k \sigma(k)}| : \sigma \  \text{is a bijection from} \ K \ \text{to}\  L  \Big\}
\] 
is reached at least by two distinct bijections $\sigma^*$ and $\tau^*$. If $\sigma^*(m+1) = \tau^*(m+1)$, this immediately shows that the $(K \setminus \{m+1\}) \times (L \setminus \{\sigma^*(m+1)\})$-submatrix of $A$ is degenerate. Thus, we can suppose that $\sigma^*(m+1)$ and $\tau^*(m+1)$ are distinct, for instance $\sigma^*(m+1) > \tau^*(m+1)$. However, as $\eps < 0$, we have:
\[
\bigodot_{k \in K} |M_{k \sigma^*(k)}| \leq \sigma^*(m+1) \eps + n \max_{ij} |A_{ij}| < \tau^*(m+1) \eps + n \min_{ij} |A_{ij}| \leq \bigodot_{k \in K} |M_{k \sigma^*(k)}| \ .
\]
In any case, we get a contradiction. We deduce that $M$ is generic.
\end{proof}
Note that, under the assumptions of Lemma~\ref{lemma:eps}, if we choose $\bo A \in \sval^{-1}(A)$ and $\bo u \in \sval^{-1}(u)$, no minor of the matrix $\begin{psmallmatrix} \bo A \\ \transpose{\bo u} \end{psmallmatrix}$ is null. This ensures that the application of the classical shadow-vertex pivoting rule also makes sense in the $\PCBC$ algorithm.

Thanks to the compatibility of the tropical shadow-vertex rule with its classical counterpart (Theorem~\ref{th:gs_trop}), we immediately obtain the following result:
\begin{proposition}\label{prop:trop_pcbc}
Let $A \in \strop^{m \times n}$, $b \in \strop^{m}$ such that $(A \ b)$ is strongly non-degenerate, and let $\eps < n (\min_{ij} |A_{ij}| - \max_{ij} |A_{ij}|)$. 

Then, the algorithm $\tropPCBC$ correctly determines whether the tropical polyhedron $\{x \in \trop^n : A^+ \tdot x \tplus b^+ \geq A^- \tdot x \tplus b^-\}$ is empty. 

Moreover, for all $\bo A \in \sval^{-1}(A)$, $\bo b \in \sval^{-1}(b)$ and $\bo \eps \in \sval^{-1}(\eps)$, the total number of basic points visited by $\tropPCBC(A, b)$ and by $\PCBC(\bo A, \bo b)$ are equal.
\end{proposition}

\begin{proof}
First note that the conditions of Theorem~\ref{th:gs_trop} are satisfied, thanks to Lemma~\ref{lemma:eps}. We are going to show by induction that the algorithms $\PCBC$ and $\tropPCBC$ iterate over the same sequence of bases. 
Initially, at Line~\lineref{line:init}, both algorithms start from the basis $(\emptyset, [n])$. 

Now, consider the $k$-th iteration of the loop between Lines~\lineref{line:bloop} and~\lineref{line:eloop}. By induction hypothesis, the points $\bar{\bo x}$ and $\bar{x}$ are basic points of $\Pbcal^{(k-1)}$ and $\Pcal^{(k-1)}$ respectively, associated with the same basis $(I_0,J_0)$. The point $\bar{\bo x}$ (resp.\ $\bar{x}$) satisfies the condition $\bo A_k \bar{\bo x} + \bo b_k \geq 0$ (resp.\ $A^+_k \tdot \bar{x} \tplus b^+_k \geq A^-_k \tdot \bar{x} \tplus b^-_k$) if, and only if, $(I_0,J_0)$ is a feasible basis of $\Pbcal^{(k)}$ (resp.\ $\Pcal^{(k)}$). As the polyhedra $\Pbcal^{(k)}$ and $\Pcal^{(k)}$ have the same feasible bases (see~\cite[Proposition~17]{AllamigeonBenchimolGaubertJoswig2013}), we deduce that the two conditions $\bo A_k \bar{\bo x} + \bo b_k \geq 0$ and $A^+_k \tdot \bar{x} \tplus b^+_k \geq A^-_k \tdot \bar{x} \tplus b^-_k$ are equivalent. 

If none of these conditions is satisfied, the two algorithms $\PCBC$ and $\tropPCBC$ execute Line~\lineref{line:trop_simplex} and run the classical and tropical shadow-vertex algorithms. Assume that the two latter algorithms are located at basic points $\bo x$ and $x$ of $\Pbcal^{(k-1)}$ and $\Pcal^{(k-1)}$ respectively, associated with the same basis $(I,J)$. Note that $\bo x$ is the final point of the path followed by the classical simplex algorithm if, and only if, $x$ is the final point of the path followed by the tropical simplex algorithm. Indeed, these two statements amount to $\sigmags(\bo A_{[k-1]}, \bo u, \transpose{\bo A_k})(I,J) = \None$ and $\sigmagstrop(A_{[k-1]}, u, \transpose{A_k})(I,J) = \None$ respectively. The equivalence then follows from Theorem~\ref{th:gs_trop}.

If $\bo x$, or equivalently $x$, is not the final point of the path, the pivoting rules $\sigmags$ and $\sigmagstrop$ returns the same leaving variable $\vout \in I \disjcup J$, still by Theorem~\ref{th:gs_trop}. In this case, the classical (resp. tropical) simplex algorithm pivots along the edge $\Ebcal$ (resp. $\Ecal$) formed by the points of $\Pbcal^{(k-1)}$ (resp. $\Pcal^{(k-1)}$) which activate the inequalities indexed by $l \in (I \disjcup J) \setminus \{\vout\}$. Let $I' \subset [m]$ and $J' \subset [n]$ such that  $I' \disjcup J' = (I \disjcup J) \setminus \{\vout\} \cup \{k\}$. 

The existence of a point $\bar{\bo x}' \in \Ebcal$ (resp.\ $\bar{x}' \in \Ecal$) such that $\bo A_k \bar{\bo x}' + \bo b_k = 0$ (resp.\ $A^+_k \tdot \bar{x}' \tplus b^+_k = A^-_k \tdot \bar{x}' \tplus b^-_k$) is equivalent to the fact that the basis $(I',J')$ is a feasible basis of $\Pbcal^{(k)}$ (resp.\ of $\Pcal^{(k)}$). Using again the fact that the polyhedra $\Pbcal^{(k)}$ and $\Pcal^{(k)}$ have the same feasible bases, we deduce that the classical simplex algorithm finds a point $\bar{\bo x}'$ such that $\bo A_k \bar{\bo x}' + \bo b_k = 0$ when pivoting along the edge $\Ebcal$ if, and only if, the tropical simplex algorithm discovers a point $\bar{x}' \in \Ecal$ which satisfies $A^+_k \tdot \bar{x}' \tplus b^+_k = A^-_k \tdot \bar{x}' \tplus b^-_k$. In this case, $\bar{\bo x}'$ and $\bar{x}'$ are basic points (of $\Pbcal^{(k)}$ and $\Pcal^{(k)}$ respectively) associated with the same basis $(I',J')$. If no such points $\bar{\bo x}'$ and $\bar{x}'$ exist, the edge $\Ebcal$ is necessarily bounded (see Footnote~\ref{footnote:ray}). Thus, the tropical edge $\Ecal = \val(\Ebcal)$ is also bounded. As a result, the two simplex algorithms both reach new basic points of $\Pbcal^{(k-1)}$ and $\Pcal^{(k-1)}$ respectively. These points are associated with the same basis $(I'',J'')$, given by $I'' \disjcup J'' = (I \disjcup J) \setminus \{\vout\} \cup \{\vin\}$ for some $\vin \not \in I \disjcup J$ (the entering variables in the classical and tropical cases are necessarily identical, by unicity).

This completes the proof by induction, and shows the second part of the proposition. The correctness of the algorithm $\tropPCBC$ immediately follows from the correctness of $\PCBC$, and the fact that $\Pbcal = \emptyset$ if, and only if, $\Pcal = \emptyset$.
\end{proof}

\begin{remark}\label{remark:eps}
As stated in Proposition~\ref{prop:trop_pcbc}, the scalar $\eps$ is supposed to be small enough. We point out that there is no need to choose or manipulate $\eps$ explicitly in the algorithm $\tropPCBC$. Indeed, as shown in the proof of Theorem~\ref{th:gs_trop}, $\eps$ is only involved in optimal assignment problems which arise in the computation of the leaving variable returned by the tropical shadow-vertex rule $\sigmagstrop(A_{[k-1]}, u, \transpose{(A_k)})$ at Line~\lineref{line:trop_simplex} ($k \in [m]$). Let us fix $k$, and let $M := \begin{psmallmatrix} A_{[k-1]} \\ \transpose{u} \\ A_k \end{psmallmatrix} \in \strop^{(k+1) \times n}$. The optimal assignment problems to be solved are associated to weighted bipartite graphs, with node sets $K \subset [k+1]$ and $L \subset [n]$ and weights $|M_{hl}|$ for $h \in K$ and $l \in L$. Let $G$ be such a graph. It involves weights with a dependency on $\eps$ if, and only if, the set $K$ contains the node $k$, which corresponds to the index of the row vector $\transpose{u}$ in the matrix $M$. More precisely, the arcs whose weight depends on $\eps$ are precisely of the form $(k,l)$, for all $l \in L$. Their respective weights are $\eps^{\tdot l} = l \times \eps$. It is clear that for any sufficiently small $\eps < 0$, the (unique) optimal assignment $\sigma^*$ in $G$ is obtained by mapping the node $k$ to the smallest element $l^*$ of $L$, and then by solving the optimal assignment problem in the graph $G'$ obtained from $G$ by removing the nodes $k$ and $l^*$ and their incident arcs. Since the weights of $G'$ do not depend on $\eps$, the optimal assignment problem in $G'$ can be solved numerically, in a standard way. 

In other words, the dependency on $\eps$ in the optimal assignment problems can be handled in a symbolic way. This is comparable to the ``lexicographic'' treatment of the scalar $\bm \eps$ in the $\PCBC$ algorithm described in~\cite[Section~6.1]{AdlerKarpShamir87}.
\end{remark}

\subsection{Average-case analysis} 

Given $\bo A \in \R^{m \times n}$, $\bo b \in \R^m$ such that no minor of the matrix $(\bo A  \  \bo b)$ is null, the probabilistic analysis of~\cite{AdlerKarpShamir87} applies to polyhedra of the form 
\[
\Pbcal_{S, S'}(\bo A, \bo b) = \{ \bo x \in \R^n : \bo x \geq 0, \ (S \bo A S') \bo x + S \bo b \geq 0 \} \enspace, 
\]
where $S = \diag(s_1, \dots, s_m)$, $S' = \diag(s'_1, \dots, s'_n)$, and the $s_i$ and $s'_j$ are i.i.d.\ entries with values in $\{\pm 1\}$ such that each of them is equal to $+1$ (resp.\ $-1$) with probability $1/2$. Equivalently, the $2^{m+n}$ polyhedra of the form $\Pbcal_{S, S'}(\bo A, \bo b)$ are equiprobable. 

\begin{theorem}[{\cite{AdlerKarpShamir87}}]\label{th:adler}
For any fixed choice of $\bo A$ and $\bo b$ such that no minor of $(\bo A  \  \bo b)$ is null, provided that $\bo \eps$ is sufficiently small, the total number of basic points visited by $\PCBC(S \bo A S',S \bo b)$ is bounded by $O(\min(m^2, n^2))$ on average. 
\end{theorem}

It can be verified that the proof of Theorem~\ref{th:adler} is still valid when we replace $\R$ by any real-closed field $K$. Alternatively, it can be shown that Theorem~\ref{th:adler} can be expressed as a first-order sentence, so that Tarski's principle can be used to transpose it to any real-closed field. As a consequence of Proposition~\ref{prop:trop_pcbc}, the algorithm $\tropPCBC$ should visit a quadratic number of tropical basic points on average. This is the way we translate the result of Adler~\etal{} to mean payoff games. 
The probability distribution of games is expressed over their payments matrices $A, B$, and must satisfy the following requirements:
\begin{assumption}\label{ass:model}
\begin{enumerate}[(i)]
\item\label{item:cond3} for all $i \in [m]$ (resp.\ $j \in [n]$), the distribution of the matrices $A, B$ is invariant by the exchange of the $i$-th row (resp.\ $j$-th column) of $A$ and $B$.
\item\label{item:cond1} almost surely, $A_{ij}$ and $B_{ij}$ are distinct and not equal to $\zero$ for all $i \in [m]$, $j \in [n]$. In this case, we introduce the signed matrix $W = (W_{ij}) \in \strop^{m \times n}$, defined by $W_{ij} := A_{ij}$ if $A_{ij} > B_{ij}$, and $\tminus B_{ij}$ if $A_{ij} < B_{ij}$.
\item\label{item:cond2} almost surely, the matrix $W$ is strongly non-degenerate.
\end{enumerate}
\end{assumption}
Let us briefly discuss the requirements of Assumption~\ref{ass:model}.

Condition~\eqref{item:cond3} corresponds to the flip invariance property. It handles discrete distributions (see Figure~\ref{fig:flip}) as well as continuous ones. In particular, if the distribution of the payment matrices admits a density function $f$, Condition~\eqref{item:cond3} can be expressed as the invariance of $f$ by exchange operations on its arguments. For instance, if $m = 1$ and $n = 2$, the flip invariance holds if, and only if, for almost all  $a_{ij}, b_{ij}$,
$f(a_{1,1}, a_{1,2}, b_{1,1}, b_{1,2}) = 
f(b_{1,1}, b_{1,2}, a_{1,1}, a_{1,2}) = 
f(b_{1,1}, a_{1,2}, a_{1,1}, b_{1,2}) =
f(a_{1,1}, b_{1,2}, b_{1,1}, a_{1,2})$.

The requirements $A_{ij}, B_{ij} \neq \zero$ for all $i,j$ in Condition~\eqref{item:cond1} ensure that the flip operations always provide games in which the two players have at least one action to play from every position. The matrix $W$ can be intuitively thought of as a tropical subtraction ``$A \tminus B$'', and the conditions $A_{ij} \neq B_{ij}$ ensure that $W$ is well defined. Then, the following result holds: 
\begin{lemma}\label{lemma:trop_pcbc_game}
If $A_{ij} \neq B_{ij}$ for all $i,j$, and $W$ is defined as in Condition~\eqref{item:cond1} of Assumption~\ref{ass:model}, then the initial state $n$ is winning in the game with matrices $A, B$ if, and only if, $\tropPCBC(W_{[m] \times [n-1]}, W_{[m] \times \{n\}})$ returns ``Non-empty''.
\end{lemma}

\begin{proof}
Given $a, b, c, d \in \trop$ such that $a \neq c$, it can be easily proved that the inequality $\max(a + x_1, b) \geq \max(c+x_1, d)$ over $x_1$ is equivalent to $b \geq \max(c+x_1,d)$ if $a < c$, and $\max(a+ x_1, b) \geq d$ if $a > c$. Using this principle, we deduce that the two systems $A_{[m] \times [n-1]} \tdot x \tplus A_{[m] \times \{n\}} \geq B_{[m] \times [n-1]} \tdot x \tplus B_{[m] \times \{n\}}$ and $W_{[m] \times [n-1]}^+ \tdot x \tplus W_{[m] \times \{n\}}^- \geq W_{[m] \times [n-1]}^- \tdot x \tplus W_{[m] \times \{n\}}^-$ are equivalent. As a result, by~\cite[Theorem~3.2]{AkianGaubertGuterman2012}, the initial state $n$ is winning if, and only if, the tropical polyhedron defined by the latter system is non-empty. This provides the expected result, thanks to the first part of Proposition~\ref{prop:trop_pcbc}.
\end{proof}

Finally, Condition~\eqref{item:cond2} is the tropical counterpart of the non-degeneracy assumption used in~\cite{AdlerKarpShamir87} to establish the average-case complexity bound. 

We point out that the set of matrices $A, B$ which do not satisfy the requirements stated in Conditions~\eqref{item:cond1} and~\eqref{item:cond2} has measure zero. As a consequence, these two conditions do not impose important restrictions on the distribution of $A, B$, and they can rather be understood as genericity conditions.

We are now ready to establish our polynomial bound on the average-case complexity of mean payoff games. 
\begin{theorem}\label{th:game_complexity}
Under a distribution satisfying Assumption~\ref{ass:model}, $\tropPCBC$ determines in polynomial time on average whether an initial state is winning for Player Max in the mean payoff game with payment matrices $A, B$. 
\end{theorem}

\begin{proof}
Without loss of generality, we assume that the initial state is the circle node $n$. 

Let us fix two payment matrices $A, B$ satisfying Conditions~\eqref{item:cond1} and~\eqref{item:cond2} of Assumption~\ref{ass:model}, and let $W$ as defined in Condition~\eqref{item:cond1}. Starting from the pair $(A, B)$ of matrices, the successive applications of row/column exchange operations precisely yield $2^{m+n-1}$ different pairs of matrices. In particular, without loss of generality, we can assume that the $n$-th columns of $A$ and $B$ have not been switched. Then, the pair of matrices that we obtained are of the form $(A_{s,s'}, B_{s, s'})$, where $s \in \{\unit, \tminus \unit\}^m$, $s' \in \{\unit, \tminus \unit\}^{n-1}$, and $A_{s,s'}$ and $B_{s,s'}$ are the matrices obtained from $A$ and $B$ respectively, by exchanging the rows $i$ and the columns $j$ such that $s_i = \tminus \unit$ and $s'_j = \tminus \unit$. The $(i,j)$-entries of $A_{s,s'}$ and $B_{s,s'}$ are distinct, and so we can define a matrix $W_{s,s'}$ in the same way we have built $W$ from $A$ and $B$. Observe that $(W_{s,s'})_{[m] \times [n-1]} = S \tdot W_{[m] \times [n-1]} \tdot S'$ and $(W_{s,s'})_{[m] \times \{n\}} = S \tdot W_{[m] \times \{n\}}$, where $S := \tdiag(s_1, \dots, s_m)$ and $S' := \tdiag(s'_1, \dots, s'_{n-1})$. Thus, by Lemma~\ref{lemma:trop_pcbc_game}, the $2^{m+n-1}$ games obtained by the successive flipping operations can be solved by calling the algorithm $\tropPCBC(S \tdot W_{[m] \times [n-1]} \tdot S', S \tdot W_{[m] \times \{n\}})$. 

Let $\bo W \in \sval^{-1}(W)$ be a fixed lift of $W$. Thanks to Condition~\eqref{item:cond2}, no minor of $\bo W$ is null. Besides, as explained in Remark~\ref{remark:eps}, we do not explicitly manipulate the scalar $\eps$ in the algorithm $\tropPCBC$. Instead, we use a symbolic technique which simulates the behavior of the tropical shadow-vertex rule for any choice of $\eps$ small enough. This ensures that for all $\bm \eps \in \sval^{-1}(\eps)$, the conditions of Theorem~\ref{th:adler} are satisfied on the instance $(\bo W_{[m] \times [n-1]}, \bo W_{[m] \times \{n\}})$. Note that the $2^{n+m-1}$ instances  $(\bo S \bo W_{[m] \times [n-1]} \bo S', \bo S \bo W_{[m] \times \{n\}})$, where $\bo S, \bo S'$ are diagonal matrices with diagonal coefficients in $\{\pm 1\}$, are respectively lifts of the instances $(S \tdot W_{[m] \times [n-1]} \tdot S', S \tdot W_{[m] \times \{n\}})$. Thanks to the second part of Proposition~\ref{prop:trop_pcbc} and Theorem~\ref{th:adler}, it follows that the total number of visited basic points when calling the algorithm $\tropPCBC(S \tdot W_{[m] \times [n-1]} \tdot S', S \tdot W_{[m] \times \{n\}})$ for all $s \in \{\unit, \tminus \unit\}^m$, $s' \in \{\unit, \tminus \unit\}^{n-1}$ is bounded by $O(2^{m+n-1} \min(m^2, n^2))$. Moreover, every iteration of the tropical simplex algorithm at Line~\lineref{line:simplex} of $\tropPCBC$ consists in determining the leaving variable returned by the tropical shadow-vertex rule, and pivoting to the next basis or computing the point $\bar{x}'$. The complexity of the former step is $O(n^3)$ by the second part of Theorem~\ref{th:gs_trop}, and the complexity of the latter step is $O(m n(m + n^2))$. Hence, every iteration is polynomial time. In total, solving the $2^{m+n-1}$ games associated with the matrices $(A_{s,s'}, B_{s,s'})$ can be done in time $O(2^{m+n-1} m n(m + n^2) \min(m^2, n^2))$. 

Let $T$ be the random variable corresponding to the time complexity of our method to solve the game with payment matrices $A,B$ drawn from a distribution satisfying Assumption~\ref{ass:model}. Similarly, given $s \in \{\unit, \tminus \unit\}^m$, $s' \in \{\unit, \tminus \unit\}^{n-1}$, let $T_{s,s'}$ be the random variable representing the time complexity to solve the game with matrices $A_{s,s'}, B_{s,s'}$, where $A, B$ are drawn from the latter distribution. Thanks to Condition~\eqref{item:cond3}, $\mathbb{E}[T] = \mathbb{E}[T_{s,s'}]$ for all $s,s'$, and so:
\[
\mathbb{E}[T] = \frac{1}{2^{m+n-1}} \times \mathbb{E} \Big[\sum_{s,s'} T_{s,s'}\Big] \leq \frac{1}{2^{m+n-1}} \times (K 2^{m+n-1} m n(m + n^2) \min(m^2, n^2))
\]
for a certain constant $K > 0$. This concludes the proof.
\end{proof}

\section{Conclusion}

We have defined a tropical analogue of the shadow-vertex simplex algorithm, and shown that every iteration has polynomial time complexity. As a corollary, we have established a polynomial-time average-case result on mean payoff games, based on the analysis of Adler, Karp and Shamir of the classical shadow-vertex algorithm. 
The main restriction of the model is the flip invariance property. It is an open question whether the tropical approach can be applied with other probabilistic models. In particular, it would be interesting to transfer smoothed complexity results, \eg~\cite{SpielmanTeng2004}, to the tropical setting.
The results of Section~\ref{sec:shadow_vertex} also suggest that there is a general method to tropicalize any semi-algebraic pivoting rule, based on the characterization of the Newton polytopes involved. This will be addressed in a future work.

\subsubsection*{Acknowledgments.} The authors thank the anonymous reviewers for their helpful comments which contributed to improve the presentation of the paper. We also thank Michael Joswig for many insightful discussions.

\bibliographystyle{alpha}
\newcommand{\etalchar}[1]{$^{#1}$}

\appendix

\end{document}